\documentclass[journal,10pt]{IEEEtran}

\usepackage{graphicx,epic,eepic,epsfig,amsmath,latexsym,amssymb,verbatim,subfigure,color}
\usepackage{multirow}
\usepackage{theorem}
\newcommand{\tr}{{\rm Tr}}

\newtheorem{definition}{Definition}

\newtheorem{lemma}[definition]{Lemma}

\newtheorem{theorem}[definition]{Theorem}

\def\ef{\mathbb{F}}

\def\squareforqed{\hbox{\rlap{$\sqcap$}$\sqcup$}}
\def\qed{\ifmmode\squareforqed\else{\unskip\nobreak\hfil
\penalty50\hskip1em\null\nobreak\hfil\squareforqed
\parfillskip=0pt\finalhyphendemerits=0\endgraf}\fi}
\def\endenv{\ifmmode\;\else{\unskip\nobreak\hfil
\penalty50\hskip1em\null\nobreak\hfil\;
\parfillskip=0pt\finalhyphendemerits=0\endgraf}\fi}
\newenvironment{remark}{\noindent \textbf{{Remarks~}}}{\qed}

\begin{document}
\title{All the stabilizer codes of distance $3$}

\author{Sixia Yu, Juergen Bierbrauer, Ying Dong, Qing Chen, and C.H. Oh
\thanks{Sixia Yu and Qing Chen, and C.H. Oh are with Center for quantum technologies, National University of Singapore, Singapore}
\thanks{Sixia Yu, Ying Dong, Qing Chen are with the Department of Modern Physics, University of Science and technology of China, China}
\thanks{Juergen Bierbrauer is with the Department of Mathematical Sciences, Michigan Technological University, Houghton, Michigan, USA}
\thanks{This work is supported by NSA grant H98230-10-1-0159 (USA),
National Research Foundation and Ministry of Education, Singapore (Grant No. WBS: R-710-000-008-271). The financial support from NSF Grant No.
No. 11075227 is also acknowledged.}}
\maketitle

\begin{abstract}
We give necessary and sufficient conditions for the existence of stabilizer codes $[[n,k,3]]$ of distance $3$ for qubits:
$n-k\ge \lceil\log_2(3n+1)\rceil+\epsilon_n$ where $\epsilon_n=1$ if
$n=8\frac{4^m-1}3+\{\pm1,2\}$ or $n=\frac{4^{m+2}-1}3-\{1,2,3\}$ for some integer $m\ge1$ and
$\epsilon_n=0$ otherwise. Or equivalently, a code $[[n,n-r,3]]$ exists if and only if
$n\leq (4^r-1)/3, (4^r-1)/3-n\notin\lbrace 1,2,3\rbrace$ for even $r$ and
$n\leq 8(4^{r-3}-1)/3, 8(4^{r-3}-1)/3-n\not=1$ for odd $r$.
Given an arbitrary length $n$ we present an explicit construction for an optimal quantum stabilizer code
of distance $3$ that saturates the above bound.
\end{abstract}

\section{Introduction}

Quantum error-correcting codes \cite{ben,knill,shor1,ste1} provide
us an active way of protecting our precious quantum data from
quantum noise and play an essential role in various quantum
informational processes. Simply speaking, a QECC is just a subspace
that corrects certain types of errors. When the subspace is specified
by the joint $+1$ eigenspace of a group of commuting multilocal Pauli
operators, i.e., direct products of local Pauli operators, the codes
are called stabilizer codes \cite{cal1,cal2,g1}. We consider only
binary codes here. As usual we shall denote by $[[n,k,d]]$ a
stabilizer code of length $n$ and distance $d$, i.e., correcting up
to $\lfloor\frac{d-1}2\rfloor$-qubit errors, that encodes $k$
logical qubits. The redundancy $r=n-k$ counts the number of the independent generators of the stabilizer.

One fundamental task is to construct optimal codes, e.g., codes with
largest possible $k$ with fixed $n$ and $d$. In the case of $d=2$
all optimal stabilizer codes are known. In the simplest nontrivial
case $d=3$, a systematic construction for all lengths has not been
achieved yet. Known results include Gottesman's optimal codes family
\cite{g2} of lengths $2^m$ with $m\ge3$ which has been generalized
for even lengths \cite{li} by using Steane's enlargement
construction \cite{ste2} with some codes being optimal and some are
suboptimal, i.e., one logical qubit less than the quantum Hamming
bound.

A code of distance $d$ is {\em degenerate} if there are harmless
undetectable errors acting on less than $d$ qubits, i.e., errors can
not be detected but do not affect the encoded quantum data. If all
errors acting on less than $d$ qubits can be detected, the codes are
{\em non-degenerate} or {\em pure}. For a pure code of distance $3$
all errors that occurred on at most $2$ qubits can be detected. The quantum
Hamming bound (qHB), e.g.,
\begin{equation}
r\ge s_H=\lceil \log_2(3n+1)\rceil \mbox{ (equivalently }
n\leq (2^r-1)/3)
\end{equation}
for a stabilizer code $[[n,n-r,3]]$, had been proven initially for
non-degenerate codes. It is also valid for degenerate
codes of distances $3$ and $5$ \cite{g1} and of a large enough length
as shown in \cite{ash} via the linear programming (LP) bound \cite{cal2,rains1}.
Our main result reads

\begin{theorem} Let $f_m=(4^m-1)/3.$ A stabilizer code $[[n,n-r,3]]$ exists if and only if
\begin{equation}\label{bd}
r\ge s_H+\epsilon_n
\end{equation}
where $\epsilon_n=1$ if $n=8f_m+\{\pm1,2\}$ or $n=f_{m+2}-\{1,2,3\}$
for some integer $m\ge1$ and $\epsilon_n=0$ otherwise
(equivalently: $n\leq f_{r/2}, f_{r/2}-n\notin\lbrace 1,2,3\rbrace$ for even $r,$
$n\leq 8f_{(r-3)/2}, n\not=8f_{(r-3)/2}-1$ for odd $r$).
\end{theorem}

For the definition of quantum stabilizer codes see~\cite{cal1,cal2}. The translation into
the language of finite geometries is in~\cite{jb}, see also the manuscript~\cite{Glynn}.
Here the Pauli matrices are identified with the binary pairs, $\lbrace I,X,Y,Z\rbrace =\ef_2^2,$ and
an $[[n,n-r]]$ quantum code is described by a check matrix of the stabilizer.
The defining condition is that any two generators are orthogonal with respect to the
symplectic inner product.
Each of the $n$ qubits corresponds to a pair of columns of the check matrix. Each column is a
binary $r$-tuple. The nonzero tuples are identified with the points of the $(r-1)$-dimensional
binary projective space: $\ef_2^r\setminus\lbrace 0\rbrace =PG(r-1,2).$ In this setting the
stabilizer is described by a family of $n$ lines in $PG(r-1,2).$
\par
After introducing some notation and recalling known results essential
to our construction in Sec.II, we shall present a general
construction for optimal codes of arbitrary length $n> 38$ that saturates the bound Eq.(\ref{bd})
in Sec.III. In Sec.IV we shall prove the only if part by showing that the qHB cannot be attained
when $\epsilon_n=1$. In Sec. V we
shall provide explicitly some of the pure optimal codes of lengths $n< 38$,
which are essential to our general construction, using a generalization of the
code pasting method.

\section{Notations and known results}

Our construction is based on two families of pure codes and
Gottesman's stabilizer pasting \cite{g3} to build new codes from old
pure codes. As usual we denote by $X,Y,Z$ the Pauli operators and
by $I$ the identity operator. Furthermore we write
$X(n)=X_1X_2\ldots X_n$ where $X_i$ is the Pauli operator $X$
acting nontrivially on the $i$-th qubit only and use analogous expressions
for $Y(n),Z(n)$, and $I(n)$. For simplicity we shall denote by
$[n,r]$ the stabilizer of a {\em pure} stabilizer code $[[n,n-r,3]]$
while simply by $[n]$ the stabilizer an {\em optimal} pure code of
length $n$, e.g., $[5]$ stands for the perfect code $[[5,1,3]]$
whose stabilizer reads
\begin{equation}\label{5}
\begin{array}{c@{\hskip 3pt}c@{\hskip 3pt}c@{\hskip 3pt}cc}\hline\hline
 X&X&X&X&I\cr
 Z&Z&Z&Z&I\cr
 X&Y&Z&I&X\cr
 Y&Z&X&I&Z\cr
\hline\hline
\end{array},
\end{equation}
where a juxtaposition of some Pauli operators in the same row means
their direct product.

{\bf Codes family $[2^m]$ $(m\ge3)$}. The first family of codes is
the Gottesman family of optimal codes  $[[2^m,2^m-m-2,3]]$ with $m\ge
3$ that saturate the quantum Hamming bound \cite{g2}. In the geometric setting
this is equivalent to the observation that the points in $PG(r-1,2)$ outside
a subspace $PG(r-3,2)$ can be partitioned into lines. In Ref.\cite{jb} this is
referred to as the Blokhuis-Brouwer construction~\cite{BE}.
By construction, these codes are non-degenerate and two observables
$X(2^m)=X_1\ldots X_{2^m}$ and $Z({2^m})=Z_1\ldots Z_{2^m}$ are
generators of the stabilizer. For simplicity we denote by $[2^m]$ a
set of $m+2$ generators of the stabilizer of Gottesman's code,
the first two generators being $X({2^m})$ and $Z({2^m})$.

An explicit construction of the remaining $m$ generators is given
by the check matrix $[H_m|A_mH_m]$ where
$H_m=[c_0,c_1,\ldots,c_{2^m-1}]$ with the $(k+1)$-th column $c_k$
being the binary vector representing integer $k$
$(k=0,1,\ldots,2^m-1)$ and $A_m$ is any invertible  and fixed point
free $m\times m$ matrix, i.e., $A_m s\ne 0$ and $A_ms\ne s$ for all
$s\in F_2^m$. As an example, the unique code $[2^3]$ has a
stabilizer generated by
\begin{equation}
\begin{array}{c@{\hskip 3pt}c@{\hskip 3pt}c@{\hskip 3pt}c@{\hskip 3pt}c@{\hskip 3pt}c@{\hskip 3pt}c@{\hskip 3pt}c}
\hline\hline X&X&X&X&X&X&X&X\cr Z&Z&Z&Z&Z&Z&Z&Z\cr
I&Z&I&Z&Y&X&Y&X\cr I&Z&X&Y&I&Z&X&Y\cr I&Y&Z&X&Z&X&I&Y\cr
\hline\hline
\end{array}.
\end{equation}

{\bf Codes family $[8\cdot m]$ $(m\ge3)$}. The second  family of
codes are of parameters $[[8m, 8m-l_m -5,3]]$ with $l_m=\lceil
\log_2m\rceil$ as constructed in Ref.\cite{li}. One crucial
property of this family is that they are stabilized by all $X$
and all $Z$ observables  $X({8m})$ and $Z({{8m}})$. Here we shall
provide a different construction based on Gottesman's family.

We divide $8m$ qubits into $m$ blocks of $8$-qubit. The first five stabilizers
of the code are $[2^3]^{\otimes m}$ whose first two generators are
$X({{8m}})$ and $Z({8m})$. In the case of $m=3,4$ the codes are
defined in Table I. In the case of $m\ge 5$ so that $l_m\ge 3$, the
remaining $l_m$ generators of the stabilizer are obtained from
Gottesman's code $[2^{l_m}]$ by at first removing the first two
generators and then removing arbitrary $2^{l_m}-m$ qubits and
finally replacing each single-qubit Pauli operators $X,Y,$ and $Z$
in the remaining stabilizers by the corresponding $8$-qubit operators
$X(2^3),Y(2^3)$, and $Z(2^3)$ respectively. In Table I we also
present an  example in the case $m=6$.

Obviously all $l_m+5$ generators defined above are commuting with
each other. Because of the first $5$ generators of the stabilizer any
$2$-errors in the same $8$-qubit block can be detected. For any $2$ errors
in two different $8$-qubit blocks,  the last $l_m$ generators together
with the first two generators define a subcode of Gottesman's code
$[2^{l_m}]$ and therefore detect all $2$ errors
in different blocks. Thus all $2$-errors can be detected so that we have constructed a pure
$1$-error-correcting code of length $8m$.

We shall abuse the notation slightly to denote all the codes of this
family by $[8\cdot m]$ though some of them are not optimal. In fact
when $f_{r+1}+1\le m\le 2^{2r+1}$ and $\frac{2^{2r+1}+1}3\le m\le
2^{2r}$ with $r\ge 1,$ the code $[8\cdot m]$ is optimal since
$l_m+5=s_H$ in these cases. Otherwise the code is suboptimal, i.e.,
$l_m+5=s_H+1$.

\begin{table}[t]\label{k}
\caption{Some examples from codes family $[8\cdot m]$.}
\begin{equation*}
\begin{array}[t]{lr}
\begin{array}{c@{\hskip 3pt}c@{\hskip 3pt}c}      \hline\hline
[2^3]&[2^3]&[2^3]\\
I(2^3)&X(2^3)&Y(2^3)\\
I(2^3)&Y(2^3)&Z(2^3)\\
\hline\hline \multicolumn{3}{c}{[8\cdot 3]=[[24,17,3]]}
\end{array}&
\begin{array}{c@{\hskip 3pt}c@{\hskip 3pt}c@{\hskip 3pt}c@{\hskip 3pt}c}
\hline\hline
[2^3]&[2^3]&[2^3]&[2^3]\\
I(2^3)&X(2^3)&Y(2^3)&Z(2^3)\\
I(2^3)&Y(2^3)&Z(2^3)&X(2^3)\\
\hline\hline\multicolumn{4}{c}{[8\cdot 4]=[[32,25,3]]}
\end{array}\\ \\
\multicolumn{2}{c}{\renewcommand{\arraystretch}{1.2}
\begin{array}{c@{\hskip 3pt}c@{\hskip 3pt}c@{\hskip 3pt}c@{\hskip 3pt}c@{\hskip 3pt}c}
\hline\hline [2^3]&[2^3]&[2^3]&[2^3]&[2^3]&[2^3]\cr
I(2^3)&Z(2^3)&Y(2^3)&X(2^3)&Y(2^3)&X(2^3)\cr
X(2^3)&Y(2^3)&I(2^3)&Z(2^3)&X(2^3)&Y(2^3)\cr
Z(2^3)&X(2^3)&Z(2^3)&X(2^3)&I(2^3)&Y(2^3)\cr \hline\hline
\multicolumn{6}{c}{[8\cdot 6]=[[48,40,3]]}
\end{array}}
\end{array}
\end{equation*}
\end{table}

{\bf Stabilizer pasting} (Gottesman \cite{g3})
In the geometric setting stabilizer pasting was rediscovered
in Ref.\cite{jb} as the generalized Blokhuis-Brouwer construction.\\
Given two
non-degenerate stabilizer codes  $[n_2,s_2]=\langle
S_1,S_2,\ldots,S_{s_2}\rangle$ and $[n_1,s_1]=\langle
T_1,T_2,\ldots,T_{s_1}\rangle$ of distance $3,$ if two observables
$X({n_2})$ and $Z({n_2})$ belong to $[n_2,s_2]$, say, $S_1=X({n_2})$
and $S_2=Z({n_2})$,   then the stabilizer defined in Table II
\begin{table}[t]\label{sp}
\caption{The stabilizer for the code obtained from pasting.}
\begin{equation*}
\renewcommand{\arraystretch}{1.2}
\begin{array}{cc}
\hline\hline
X({n_2})&I({n_1})\cr
Z({n_2})&I({n_1})\cr\hline
S_{3}&T_{1}\cr
S_{4}&T_{2}\cr
\vdots&\vdots\cr
S_{s_2}&T_{s_2-2}\cr\hline
I({n_2})&T_{s_2-1}\cr
\vdots&\vdots\cr
I({n_2})&T_{s_1}\cr\hline\hline
\end{array}\quad\quad \mbox{or}\quad\quad
\begin{array}{cc}
\hline\hline
X({n_2})&I({n_1})\cr
Z({n_2})&I({n_1})\cr\hline
S_{3}&T_{1}\cr
S_{4}&T_{2}\cr
\vdots&\vdots\cr
S_{s_1+2}&T_{s_1}\cr\hline
S_{s_1+3}&I({n_1})\cr
\vdots&\vdots\cr
S_{s_2}&I({n_1})\cr\hline\hline
\end{array}
\end{equation*}
\end{table}
defines a non-degenerate stabilizer code $[n_2+n_1,s]$ with
$s=\max\{s_2,s_1+2\}$, denoted as $[n_2,s_2]\rhd[n_1,s_1]$.

As a first example of stabilizer pasting we obtain an optimal
code $[13]=[[13,7,3]]$ by pasting the optimal code $[2^3]$ of length
$n_2=8$ and $s_2=5$ stabilizers with the perfect code $[5]$, i.e,,
$n_1=5$ and $s_1=4$. The resulting code is of length $n_1+n_2=13$
with $s_1+2=6> s_2=5$ stabilizers.

If there is a third pure code $[n_3,s_3]$ with $X(n_3)$ and $Z(n_3)$
belonging to its stabilizer then the stabilizer pasting results in a
pure code
\begin{equation}
[n_1+n_2+n_3,s]=[n_3,s_3]\rhd[n_2,s_2]\rhd[n_1,s_1]
\end{equation}
with $s=\max\{s_3,s_2+2,s_1+4\}$, which can be further pasted with
another code and so on. As a second example the perfect code
$[[f_m,f_m-2m,3]]$ with $f_m=\frac{4^m-1}3$ and $m\ge 3$ can be
constructed by pasting Gottesman's codes $[2^{2l}]$
$(l=2,3,\ldots,m)$ with the pure perfect $5$-qubit code
\cite{g3,cal2},
\begin{equation}\label{p1}
[f_{m}]=[2^{2(m-1)}]\rhd[2^{2(m-2)}]\rhd\ldots\rhd[2^4]\rhd[5].
\end{equation}
As a last example the optimal stabilizer code of length $8f_m$
$(m\ge 2)$ can be constructed by pasting Gottesman's codes
$[2^{2l+1}]$ $(l=1,3,\ldots,m)$ \cite{cal2}
\begin{equation}\label{p2}
[8f_m]=[2^{2m+1}]\rhd[2^{2m-1}]\rhd\ldots\rhd[2^3].
\end{equation}

\section{The general construction}
Our main tool is the pasting of codes to produce new codes from old
ones. Only pure codes can be used in the pasting. Since the
optimal stabilizer code for $n=6$ is degenerate we see that optimality does
not imply pureness. Although from the Grassl's public code table  we
know that the optimal codes for $n\le 37$ exist, we need to check in each case
that pure optimal codes exist.

\begin{table}[t]\label{37}
\renewcommand{\arraystretch}{1.2}
\caption{The stabilizers of the pure optimal codes $[[n,n-r,3]]$ for
$n\le 38$ and $n\ne 6$. All the $2$-error-detecting blocks such as $[28,7]_2$
are constructed in Sec. V explicitly.}
$$
\begin{array}{r@{\hskip 4pt}clr@{\hskip 4pt}cl}
\hline\hline
n &r& \mbox{Stabilizer} &             n&r& \mbox{Stabilizer}\\\hline
10&6&\mbox{Table VI}&                5&4& [4,4]_1\rhd[1]_1\\
11&6&[10,6]_1\rhd[1]_1&               7&6&[6,6]_1\rhd[1]_1\\
12&6&[10,6]_2\rhd[2,4]_2&             8&5&[2^3]\\
13&6&[10,6]_2\rhd[3,4]_2&             9&6&[6,6]_2\rhd[3,4]_2\\
14&6&[10,6]_1\rhd[4,4]_1&            18&7&[10]\rhd[2^3]\\
15&6&[10]\rhd[5]&                             19&7&[18,7]_1\rhd[1]_1\\
16&6&[2^4]&                           20&7&[18,7]_2\rhd[2,4]_2\\
17&6&\mbox{Eq.(\ref{17})} &           21&6&[2^4]\rhd[5]\\
&&&     22&7&[18,7]_1\rhd[4,4]_1\\
31&7&[28,7]_2\rhd[3,4]_2&             23&7&[18,7]_2\rhd[5,5]_2\\
32&7&[2^5]&                           24&7&[8\cdot 3]\\
33&7&[28,7]_2\rhd[5,5]_2&             25&7&[18,7]_1\rhd[7,5]_1\\
34&7&[26,7]_2\rhd[7,5]_1\rhd[1]_1&    26&7&[18,7]_2\rhd[7,5]_1\rhd[1]_1\\
35&7&[28,7]_1\rhd[7,5]_1&             27&7&[18,7]_1\rhd[2^3]\rhd[1]_1\\
36&7&[28,7]_2\rhd[7,5]_1\rhd[1]_1&    28&7&[20,7]_2\rhd[7,5]_1\rhd[1]_1\\
37&7&[32]\rhd[5]&                     29&7&[8\cdot 3]\rhd [5]\\
38&7& \mbox{Eq.(\ref{38})} &          30&7&[28,7]_2\rhd[2,4]_2\\
\hline\hline
\end{array}
$$
\end{table}

\begin{lemma}
Non-degenerate optimal $1$-error correcting codes of lengths $10\le
n\le 17$ and $31\le n\le 38$ exist.
\end{lemma}
\begin{proof}
An example of a pure code $[[17,11,3]]$ was found in Ref.\cite{cal2} by a random search. A
geometric construction in Ref.\cite{jb} yields the following set of generators of the stabilizers
\begin{equation}\label{17}
\begin{array}{c@{\hskip 3pt}c@{\hskip 3pt}c@{\hskip 3pt}c@{\hskip 3pt}c@{\hskip 3pt}
c@{\hskip 3pt}c@{\hskip 3pt}c@{\hskip 3pt}c@{\hskip 3pt}c@{\hskip 3pt}c@{\hskip 3pt}
c@{\hskip 3pt}c@{\hskip 3pt}c@{\hskip 3pt}c@{\hskip 3pt}c@{\hskip 3pt}c}
\hline\hline
X&X&I&  Y&X&X& Y&Y& X&X&X& I&I&I& Z&Z&Z\cr
X&Y&I&  X&Y&X& X&Z& Z&Z&Z& X&X&X& I&I&I\cr
X&Z&I&  X&X&Y& Z&X& I&I&I& Z&Z&Z& X&X&X\cr
Z&I&X&  Y&Z&Z& X&Z& Y&I&Z& Y&I&Z& Y&I&Z\cr
Z&I&Y&  Z&Y&Z& Y&X& Z&Y&I& Z&Y&I& Z&Y&I\cr
Z&I&Z&  Z&Z&Y& Z&Y& I&Z&Y& I&Z&Y& I&Z&Y\cr
\hline\hline
\multicolumn{17}{c}{[17]=[[17,11,3]]}
\end{array}.
\end{equation}
A direct application of stabilizer pasting
to two optimal codes yields an optimal pure
code $[[37,30,3]]$ whose stabilizer reads
$
[37]=[2^5]\rhd [5].
$

The following construction of an optimal pure code $[38]=[[38,31,3]]$ is translated
from the geometric construction in Ref.\cite{jb}.
We denote by $H_{32}=[H_{26},A,B]$ a $5\times 2^5$ matrix whose columns $h_i$
are all possible $5$-dimensional vector with entries $0,1$ where $A,B$ are two $5\times 3$ matrices
\begin{equation}
A=\left(
\begin{array}{ccc}
0&1&0\cr
0&0&0\cr
0&0&1\cr
0&0&0\cr
0&0&1
\end{array}\right),\quad
B=\left(
\begin{array}{ccc}
1&1&1\cr
0&1&1\cr
0&0&1\cr
1&0&1\cr
0&1&0
\end{array}
\right).
\end{equation}
Also we denote by $H^\prime_{32}=[H^\prime_{26},A^\prime,B^\prime]=E_1+MH_{32}$ which is
another $5\times 2^5$ matrix, where
\begin{equation}
E_1=\left(\begin{array}{ccccc}
1_{32}\cr
0_{32}\cr0_{32}\cr0_{32}\cr0_{32}
\end{array}\right),\quad M=\left(\begin{array}{ccccc}
0 & 1 & 0 & 1 & 1\cr
1 & 0 & 0 & 0 & 0\cr
0 & 1 & 0 & 0 & 1\cr
0 & 1 & 0 & 1 & 0\cr
0 & 0 & 1 & 0 & 1\cr
\end{array}\right).
\end{equation}
Both $M$ and $I+M$ are invertible. Furthermore we denote
\begin{equation*}
[P|Q]=\left[\begin{array}{cccccc|cccccc}
0 & 0 & 1 & 1 & 1 &1 &1 & 0 & 0 & 0 & 0 &0\cr
0 & 0 & 0 & 0 & 0 &0 &0 & 1 & 1 & 1 & 1 & 1\cr
0 & 0 & 0 & 1 & 1 &0&1 & 1 & 1 & 0 & 0 &1\cr
1 & 0 & 0 & 0 & 1 &1&0 & 1 & 0 & 1 & 1 &0\cr
0 & 1 & 1 & 0 & 0 &1&1 & 0 & 1 & 0 & 1 &0\cr
\end{array}\right].
\end{equation*}
The check matrix of the stabilizer $[38]$ reads
\begin{equation}\label{38}
\left[
\begin{array}{cccc|cccc}
1_{26}&1_{3}&1_{3}&0_6  &1_{26}&1_{3}&1_{3}&0_6          \cr
0_{26}&0_{3}&1_{3}&0_6  &1_{26}&0_{3}&1_{3}&0_6          \cr
H_{26}&A&A^\prime&P&H_{26}^\prime&B&B^\prime&Q
\end{array}\right].
\end{equation}

Optimal pure codes of lengths $16$ and $32$ exist. We shall
postpone the explicit constructions of pure optimal codes of
the remaining lengths to Sec. V where the pasting of stabilizers is
generalized to the pasting of noncommuting sets of generators.
A typical example is the construction of an optimal pure code $[36]=[[36,29,3]]$
whose stabilizer is explicitly given in Table
V. All the pure optimal codes of lengths $5\le n\le 38$ with $n\ne
6$ are summarized in Table III.
\end{proof}

Lemma 2 ensures that there exist  $[17-\beta]$ and $[38-\beta]$ for
$0\le \beta\le 7$, i.e, optimal pure codes of those lengths exist
and have $6$ and $7$ generators respectively. For
$n> 38$ we have the following general construction:

\begin{theorem} Suppose $n> 38$ and $n\ne 8^\epsilon f_m$ for any integer $m$ and $\epsilon=0,1$.
a) If $8f_{m}-1\le n\le f_{m+2}-4$ for some $m\ge 2$ then the stabilizer
\begin{equation}\label{ca}
[8\cdot (2^{2m-1}-\alpha)]\rhd [2^{2m}]\rhd [2^{2m-2}]\rhd\ldots\rhd [2^6]\rhd[17-\beta]
\end{equation}
defines an optimal pure code $[[n,n-2m-4,3]]$ where
 $f_{m+2}-4-n=8\alpha+\beta$ with $\alpha\ge 0$ and $0\le
\beta\le 7$. When $m=2$ the
stabilizer is generated by
$[8\cdot(8-\alpha)]\rhd[17-\beta]$. b) If $f_{m+2}-3\le n
\le 8f_{m+1}-2$ for some $m\ge 2$ then the stabilizer
\begin{equation}\label{cb}
[8\cdot(2^{2m}-\alpha)]\rhd [2^{2m+1}]\rhd [2^{2m-1}]\rhd\ldots \rhd [2^7]\rhd [38-\beta]
\end{equation}
defines an optimal pure code $[[n,n-2m-5,3]]$ where
$8f_{m+1}-2-n=8\alpha+\beta$ with $\alpha\ge 0$ and $0\le \beta\le 7$.  When $m=2$ the
stabilizer is generated by
$[8\cdot(16-\alpha)]\rhd[38-\beta]$.
\end{theorem}
\begin{proof} At first from Lemma 1 and the
constructions of two codes families $[8\cdot k]$ and $[2^k]$ it is
clear that all the stabilizer codes involved in Eq.(\ref{ca}) or
Eq.(\ref{cb}) are non-degenerate. Secondly by construction two
families of codes $[8\cdot k]$ and $[2^k]$ are stabilized by all $X$
and all $Z$ Pauli operators. As a result the stabilizer pasting can
be applied from right to left so that Eq.(\ref{ca}) and
Eq.(\ref{cb}) define pure stabilizer codes of distance $3.$

Now we evaluate the parameters of the codes. It is easy to see from
the definition of $\alpha$ and $\beta$ and the identity
$f_{m+2}=2^{2m+2}+2^{2m}+\ldots+2^4+5$ that the length of the
resulting codes are exactly $n$. Recalling that the codes $[8\cdot
k]$ and $[2^k]$ have $l_k=\lceil \log k\rceil+5$ and $k+2$
stabilizers respectively  while the codes $[17-\beta]$ and
$[38-\beta]$  have at most $6$ and $7$ stabilizers respectively. Since
$\alpha\ge 0$ we have $\lceil \log(2^{2m-a}-\alpha)\rceil\le 2m-a$
for $a=0,1$, the stabilizers in Eq.(\ref{ca}) and Eq.(\ref{cb}) have
$2m+4$ and $2m+5$ generators, respectively.
\end{proof}

As a first example when $n=81$ we have
$
[81]=[2^6]\rhd[17]
$
which is an optimal code $[[81,73,3]]$ apparently missing from the public code
table. As another example when $n=305$ we have $m=3$ and $8f_3-1=167<305<f_5-4=337$
so that construction a) applies. Also we have $\alpha=4$ and $\beta=0$ and as a consequence
$[305]=[8\cdot 28]\rhd[2^6]\rhd[17]$. As a last example
$n=371$ we have $m=3$ and $340\le n\le 677$ with the
condition of case b satisfied. In this case $677-371=8\times38+2$ so
that $\alpha=38$ and $\beta=2$ and by construction Eq.(\ref{cb}) we
have $[371]=[8\cdot 26]\rhd [2^7]\rhd [35] $.
Both codes $[[305,195,3]]$ and $[[371,360,3]]$ saturate the quantum Hamming bound.

\section{Exact bound}

In this section we shall prove the `only if' part of Theorem 1, which amounts to showing that in the case of
$\epsilon_n=1$, i.e., $n=8f_m+\{\pm1,2\}$ or $n=f_{m+2}-\{1,2,3\}$ for some $m\ge1$,
the quantum Hamming bound cannot be attained. Suppose that there is a pure code $[[n,k,3]]$
that attains the quantum Hamming bound, i.e., a code whose stabilizer has $s_H$ generators.
Let $[G_x|G_z]$ be its check matrix which is an $s_H\times 2n$ matrix satisfying $G_xG^T_z+G_zG^T_x=0$.
Because the code is supposed to be pure, the matrix $S=[G_x|G_z|G_x+G_z]$, composed of the syndromes
of all possible $1$-qubit errors, must have distinct columns. Moreover we have $SS^T=0$, meaning that
$S$ is self orthogonal. Denote by  $Y$
the $s_H\times y$ matrix composed of $y=2^{s_H}-3n-1$ $s_H$-dim column vectors that are not syndromes
of any $1$-qubit errors. Being composed all possible $s_H$-dim vectors the matrix
$H_{2^{s_H}}=[0|S|Y]$ is self orthogonal and thus $Y$ is also self orthogonal.
In other words, the matrix $Y$ is the check matrix of some classical binary self-orthogonal code
$[y,k,3]_2$ for some $k\le s_H$. On the one hand it is an elementary fact that
such self-orthogonal codes exist only for $y=7,8$ when  $y\le 10$ \cite{jb}.
On the other hand in the case of $\epsilon_n=1$ we have $y\in\{1,4,10\}$ if $n=8f_m+\{\pm1,2\}$
while $y\in\{3,6,9\}$ if $n=f_{m+2}-\{1,2,3\}$.
This contradiction proves that the qHB cannot be attained by a pure code in the case $\epsilon_n=1$.

Now suppose that the code $[[n,k,3]]$ attaining the qHB is impure. In this case some generators of the
stabilizer act nontrivially only on 1) one qubit or 2) two qubits. In case 1) by removing this generator
together with the qubit it acts on we obtain a code $[[n-1,k,3]]$ which may be pure or impure.
From the qHB for the code $[[n-1,k,3]]$, i.e., $n-k-1\ge s_H(n-1)$, and $s_H(n)=s_H(n-1)$
in the case of $\epsilon_n=1$,  the  bound Eq.(\ref{bd}) follows immediately. Therefore we can assume that case 1) does
not happen.

In case 2) there are some single-qubit errors acting on different qubits that lead to an identical syndrome.
We suppose that there is a number $v\ge1$ of such degenerated syndromes with each syndrome caused by $u_i+1$ single-qubit errors (acting on different qubits since case 1 does not happen) where $u_i\ge 1$ and $i=1,2,\ldots,v$.
Because the product of two single-qubit errors that lead to the same syndrome is a stabilizer of the code,
there is a set $U$ of
 generators of the stabilizer that act nontrivially exactly on two qubits and obviously $|U|=\sum_{i=1}^vu_i:=u\le n-k$.
According to Ref.\cite{jb} (Theorem 3.2) it holds \begin{equation}\label{im}
 n-(u+v)\le \frac{2^{n-k-u}-v-1}3.
 \end{equation}

Here we provide an alternative proof of the above inequality which may apply also to nonadditive codes.
Let $W_i$ be the set of qubits that $u_i+1$ single-qubit errors, which lead to an identical syndrome, act on and obviously $|W_i|=u_i+1$ since different errors must act on different qubits.
Because two different degenerated syndromes
cannot be caused by single-qubit errors acting on the same qubit, we have a disjoint union
$W=\cup_{i=1}^v W_i$ with $|W|=u+v$. Let $\bar W$ denote the remaining $|\bar W|=n-u-v$ qubits that
all the generators in $U$ trivially act on. Without loss of generality, applying some local
Clifford transformations and relabeling the qubits when necessary, we can assume that
those $v$ degenerated syndromes are caused by single-qubit errors $X_i$ with $i=1,2,\ldots v$.
Define
\begin{equation}
\hat P_1=\hat P+\sum_{i=1}^vX_i\hat PX_i+\sum_{E_a,a\in \bar W}E_a \hat P E_a
\end{equation}
where $\hat P$ is the projector of the coding subspace of $[[n,k,3]]$ and the last summation is over
all possible $1$-qubit errors ($3|\bar W|$ of them) in qubits belonging to $\bar W$.
Note that each term in the definition of $\hat P_1$ is a projector and all these projectors are
orthogonal to each other.  Let $\hat Q$ be the projector of the subspace stabilized by the generators in
$U$ and obviously $\tr\hat Q=2^{n-|U|}$. Being also stabilized by
 $U$, the subspace $\hat P_1$ is a subspace $\hat Q$. As a consequence
$\tr \hat P_1\le \tr \hat Q$, i.e.,
$\big(1+v+3|\bar W|\big)\tr P\le {2^{n-u}}$,
which becomes  exactly the inequality  Eq.(\ref{im}) considering $\tr P=2^k$.

From inequality Eq.(\ref{im}) it follows that an impure code attaining the qHB must satisfy
$3n+1\le 2^{s_H-u}+3u+2v$ which will be shown in what follows to be impossible when $\epsilon_n=1$.
Suppose $s_H\ge 6.$ It follows from $\epsilon_n=1$ that $3n+1\ge 2^{s_H}+10$ and we shall prove
$2^{s_H}(1-2^{-u})> 10+3u+2v$. Indeed if $u\le 6$ we have always
$2^{s_H}(1-2^{-u})\ge 2^6-2^{6-u}> 10+5u\ge 10+3u+2v$. If $u\ge 6$ we have
$2^{s_H}(1-2^{-u})\ge 63\times 2^{s_H-6}> 10+5s_H\ge 10+2v+3u$ since $v\le u\le s_H$.
Suppose now $s_H=5$ and from $\epsilon_n=1$ it follws  $n=7,9,10$.
In case $n=7$ inequality Eq.(\ref{im}) becomes $22\le 2^{5-u}+3u+2v$. This is impossible because
for $u=1,2,3$ we have $22-2^{5-u}> 5u\ge 3u+2v$ and for $u\ge 4$ we have $22-2^{5-u}> 14+u\ge 3u+2v$
since $u+v\le 7$.
In the cases of $n=9,10$ the inequality Eq.(\ref{im}), which becomes $28,31\le 2^{5-u}+3u+2v$,
is impossible because $2^{5-u}+5u\le 26$ for $1\le u\le 5$.
If $s_H=4$ and $\epsilon_n=1$ we have $n=4$ and the corresponding code must be pure.
All these contradictions show that the qHB cannot be attained by impure code either
when $\epsilon_n=1$.

\section{Special constructions}

In this section we shall prove Lemma 1 by constructing explicitly
all the remaining optimal non-degenerate codes of lengths $n\le 38$ except for
$n=6$. Our main tool is a generalization of the pasting of
stabilizer codes to a pasting of $2$-error detecting blocks
(2ed-block) as defined below.

\begin{definition} A $2$-error detecting block $[n,s]_e$ is generated by a set of $s$ multilocal
Pauli operators acting on $n$ qubits with $e$ pairs being non-commuting
that detects up to $2$-qubit errors.
\end{definition}

\begin{table}[b]\label{2ed1}
\caption{Some examples of $2$-error-detecting blocks.}
\begin{equation*}
\begin{array}{c}
\begin{array}{c@{\hskip 3pt}c}\hline\hline X&I\\ Z&I\\ I&X\\ I&Z\\ \hline\hline\multicolumn{2}{c}{[2,4]_2}\end{array}\quad
\begin{array}{c@{\hskip 3pt}c@{\hskip 3pt}c}\hline\hline X&X&X\\ Z&Z&Z\\ X&Y&Z\\ Y&Z&X\\ \hline\hline\multicolumn{3}{c}{[3,4]_2}\end{array}\quad
\begin{array}{c@{\hskip 3pt}c@{\hskip 3pt}c@{\hskip 3pt}c}\hline\hline X&X&X&X\\ Z&Z&Z&Z\\ X&Y&Z&I\\ Y&Z&X&I\\ \hline\hline\multicolumn{4}{c}{[4,4]_1}\end{array}
\\ \\ \begin{array}{c@{\hskip 3pt}c@{\hskip 3pt}cc@{\hskip
3pt}c@{\hskip 3pt}c@{\hskip 3pt}c@{\hskip 3pt}c} \hline\hline
X&X&X\cr Z&Z&Z\cr Z&I&Z\cr Z&X&Y\cr Y&Z&X\cr \hline\hline
\multicolumn{3}{c}{[3,5]_2}
\end{array}\quad
\begin{array}{c@{\hskip 3pt}c@{\hskip 3pt}c@{\hskip 3pt}c@{\hskip 3pt}c@{\hskip 3pt}c@{\hskip 3pt}c}
\hline\hline X&X&X&X&X\cr Z&Z&Z&Z&Z\cr Y&X&Y&X&I\cr I&Z&X&Y&I\cr
Z&X&I&Y&I\cr \hline\hline \multicolumn{5}{c}{[5,5]_2}
\end{array}\quad\begin{array}{c@{\hskip 3pt}c@{\hskip 3pt}c@{\hskip 3pt}c@{\hskip 3pt}c@{\hskip 3pt}c@{\hskip 3pt}c}
\hline\hline X&X& X&X&X& X&X\cr Z&Z& Z&Z&Z& Z&Z\cr Z&I&Z&Y&X&Y&X\cr
Z&X&Y&I&Z&X&Y\cr Y&Z&X&Z&X&I&Y\cr \hline\hline
\multicolumn{7}{c}{[7,5]_1}
\end{array}
\end{array}
\end{equation*}
\end{table}

Each non-degenerate stabilizer code $[n,s]$ detects all 2-errors and
so they define 2ed-blocks $[n,s]_0$ with all the generators
commuting. By shortening a pure code we generally obtain 2ed-blocks
with some noncommuting pairs of generators. Some examples of
2ed-blocks are presented in Table IV.

{\bf 2ed-blocks pasting}: Given two 2ed-blocks $[n_2,s_2]_{e_2}$ and
$[n_1,s_1]_{e_1}$ that are generated by $\langle
S_1=X(n_2),S_2=Z(n_2),\ldots,S_{s_2}\rangle$ and $\langle
T_1,T_2\ldots,T_{s_1}\rangle$ respectively, then
$s=\max\{s_1,s_2+2\}$ generators as given in Table II is a 2-ed
block $[n_1+n_2,s]_e$ with $|e_1-e_2|\le e\le e_1+e_2$. For
convenience we shall denote by $[n_1,s_2]_{e_1}\rhd[n_2,s_1]_{e_2}$
the resulting 2ed-block.

The 2ed-block given in Table II detects up to 2-qubits errors
because firstly all the errors happening on the $n_1$-block or
$n_2$-block can be detected because  $[n_1,s_1]_{e_1}$ and
$[n_2,s_2]$ are two pure codes of distance 3 and secondly two qubits
errors happening on different blocks can be detected by the first
two generators $X({n_2})\otimes I(n_1)$ and $Z({n_2})\otimes
I(n_1)$. If two noncomuting generators are arranged in the same row
the resulting generators will become commuting. As a result $e$ can
be zero when $e_1=e_2$ and all noncommuting pairs are carefully
matched. In this case we obtain a pure 1-error-correcting stabilizer
code, since all 2-qubit errors can be detected.

From the above arguments we see that although the $1$-qubit block,
denoted as $[1]_1=\langle X,Z\rangle$,  detects only single qubit
errors, it can be regarded as a 2ed-block because there is no
2-qubit errors on a single qubit block. For example we have
$[2,4]_2=[1]_1\rhd[1]_1$. As another example the perfect code
$[[5,1,3]]$ in Eq.(\ref{5}) can be regarded as the pasting of two
2ed-blocks $[4,4]_1\rhd[1]_1$.

A 2ed-block fails to define a code because there are some pairs of
noncommuting generators. By pasting two or more 2ed-blocks these
noncommuting generators may become commuting and we thus obtain a
1-error correcting stabilizer code.  Our construction is therefore a
kind of puncturing plus pasting. By puncturing some old stabilizer
codes we obtain some 2ed-blocks that generally contain some  pairs
of noncommuting generators. By pasting with some other 2ed-blocks
and carefully matching their noncommuting pairs we are able to
produce some new stabilizer codes. To complete the constructions
given in Table III we have only to construct explicitly all the
relevant 2ed-blocks.

\begin{table*}[t]\label{ds}
\renewcommand{\arraystretch}{1.2}
\centering \caption{The stabilizer for the optimal code
$[[36,29,3]]$.}
\begin{tabular}{c@{\hskip 4pt}c@{\hskip 4pt} c@{\hskip 4pt}c c@{\hskip 4pt}c@{\hskip 4pt}
c@{\hskip 4pt}c@{\hskip 4pt}c@{\hskip 4pt}c@{\hskip 4pt}c@{\hskip
4pt}c@{\hskip 4pt}c@{\hskip 4pt} c@{\hskip 4pt}c@{\hskip
4pt}c@{\hskip 4pt}c@{\hskip 4pt}c@{\hskip 4pt}c@{\hskip
4pt}c@{\hskip 4pt}c@{\hskip 4pt} c@{\hskip 4pt}c@{\hskip
4pt}c@{\hskip 4pt}c@{\hskip 4pt}c@{\hskip 4pt}c@{\hskip
4pt}c@{\hskip 4pt} c@{\hskip 4pt}c@{\hskip 4pt}c@{\hskip
4pt}cc@{\hskip 4pt}c@{\hskip 4pt}c@{\hskip 4pt}c@{\hskip
4pt}c@{\hskip 4pt}c@{\hskip 4pt} cc}
\multicolumn{32}{c}{$[[32,25,3]]$} \cr \hline\hline
5&10&19&28&0&1&2&3&4&6&7&8&9&11&12&13&14&15&16&17&18&20&21&22&23&24&25&26&27&29&30&31\\\hline
$X$&$X$&$X$&$X$&$X$&$X$&$X$&$X$&$X$&$X$&$X$&$X$&$X$&$X$&$X$&$X$&$X$&$X$&$X$&$X$&$X$&$X$&$X$&$X$&$X$&$X$&$X$&$X$&$X$&$X$&$X$&$X  $&$I$&$I$&$I$&$I$&$I$&$I$&$I$&$I$\\
$Z$&$Z$&$Z$&$Z$&$Z$&$Z$&$Z$&$Z$&$Z$&$Z$&$Z$&$Z$&$Z$&$Z$&$Z$&$Z$&$Z$&$Z$&$Z$&$Z$&$Z$&$Z$&$Z$&$Z$&$Z$&$Z$&$Z$&$Z$&$Z$&$Z$&$Z$&$Z  $&$I$&$I$&$I$&$I$&$I$&$I$&$I$&$I$\\
$Z$&$I$&$X$&$Y$&$I$&$Z$&$I$&$Z$&$I$&$I$&$Z$&$I$&$Z$&$Z$&$I$&$Z$&$I$&$Z$&$Y$&$X$&$Y$&$Y$&$X$&$Y$&$X$&$Y$&$X$&$Y$&$X$&$X$&$Y$&$X  $&$X$&$X$&$X$&$X$&$X$&$X$&$X$&$I$\\
$Y$&$Z$&$Y$&$Z$&$I$&$Y$&$Z$&$X$&$I$&$Z$&$X$&$I$&$Y$&$X$&$I$&$Y$&$Z$&$X$&$Z$&$X$&$I$&$Z$&$X$&$I$&$Y$&$Z$&$X$&$I$&$Y$&$X$&$I$&$Y  $&$Z$&$Z$&$Z$&$Z$&$Z$&$Z$&$Z$&$I$\\
$Z$&$X$&$Z$&$X$&$I$&$Z$&$I$&$Z$&$I$&$I$&$Z$&$X$&$Y$&$Y$&$X$&$Y$&$X$&$Y$&$I$&$Z$&$I$&$I$&$Z$&$I$&$Z$&$X$&$Y$&$X$&$Y$&$Y$&$X$&$Y  $&$Y$&$Z$&$X$&$Z$&$X$&$I$&$Y$&$X$\\
$I$&$X$&$X$&$I$&$I$&$Z$&$X$&$Y$&$Z$&$Y$&$X$&$I$&$Z$&$Y$&$Z$&$I$&$Y$&$X$&$Z$&$I$&$Y$&$I$&$Z$&$X$&$Y$&$Z$&$I$&$Y$&$X$&$Z$&$X$&$Y  $&$Z$&$X$&$Y$&$I$&$Z$&$X$&$Y$&$Z$\\
$Y$&$Y$&$Y$&$Y$&$I$&$Z$&$I$&$Z$&$X$&$X$&$Y$&$Y$&$X$&$X$&$Z$&$I$&$Z$&$I$&$X$&$Y$&$X$&$I$&$Z$&$I$&$Z$&$Z$&$I$&$Z$&$I$&$X$&$Y$&$X  $&$Z$&$I$&$Z$&$Y$&$X$&$Y$&$X$&$I$\\
\hline\hline
&&&&\multicolumn{28}{c}{$[28,7]_2$}&\multicolumn{7}{c}{$[7,5]_1$}&$[1]_1$
\end{tabular}
\end{table*}

We consider the optimal code $[2^5]$ as in Table V whose stabilizer
is defined by the check matrix $[RH_5|A_5 RH_5]$ with
\begin{equation}
A_5=\left(\begin{array}{c@{\hskip 4pt}c@{\hskip 4pt}c@{\hskip 4pt}c@{\hskip 4pt}c}1&1&0&0&0\\1&1&0&1&0\\0&1&0&0&0\\0&1&1&0&1\\0&1&1&0&0\end{array}\right),\quad
R=\left(\begin{array}{c@{\hskip 4pt}c@{\hskip 4pt}c@{\hskip 4pt}c@{\hskip 4pt}c}1&0&0&0&0\\0&0&0&0&1\\0&1&0&0&0\\0&0&0&1&0\\1&1&1&0&0\end{array}\right).
\end{equation}
Obviously $A_5$ is revertible and fixed-point free and $R$ is
invertible. By removing four coordinates
$[c_5,c_{10},c_{19},c_{28}]$ from this $[2^5]$ we obtain the
2ed-block $[28,7]_2$ and by removing the first four coordinates
$[c_0,c_1,c_2,c_3]$ we obtain A 2ed-block $[28,7]_1$. By 2ed-blocks
pasting with 2ed-blocks in Table IV we obtain the pure optimal codes
of lengths $30, 31,33$ and $35$ in addition to a previously unknown
optimal code
\begin{equation}
[36]=[28,7]_2\rhd[7,5]_1\rhd[1]_1
\end{equation}
whose stabilizer is explicitly given in Table V.

From three partitions of $[2^4]$ as shown in Table VI we can obtain
a pure optimal code $[10]$ as well as the unique optimal code
$[[6,0,4]]$ of distance 4 and four different 2ed-blocks. By pasting
with the perfect 5-qubit code we obtain $[15]=[10]\rhd[5]$. Also we
obtain all the optimal pure codes of lengths from 11 to 14 as well
as an optimal pure $[7]=[6,6]_1\rhd[1]_1$. Finally the remaining
2ed-blocks appeared in  Table III are given in Table VII.

\begin{table}[h]\label{10}
\renewcommand{\arraystretch}{1.2}
\centering \caption{Three partitions of the optimal code $[2^4]$.}
\begin{equation*}
\begin{array}{c}
\begin{array}{c@{\hskip 3pt}c@{\hskip 3pt}c@{\hskip 3pt}c@{\hskip 3pt}c@{\hskip 3pt}
c@{\hskip 3pt}c@{\hskip 3pt}c@{\hskip 3pt}c@{\hskip 3pt}cc@{\hskip
3pt} c@{\hskip 3pt}c@{\hskip 3pt}c@{\hskip 3pt}c@{\hskip 3pt}c}
\multicolumn{16}{c}{[2^4]=[[16,10,3]]}\cr \hline\hline X& X&X&X&
X&X&X&X&X&X&       X&X&X& X&X&X\cr Z& Z&Z&Z& Z&Z&Z&Z&Z&Z&
Z&Z&Z& Z&Z&Z\cr I& X&Y&Z& I&I&I&X&Y&Z&       Y&X&Z& Z&Y&X\cr I&
Y&Z&X& I&I&I&Y&Z&X&       Z&Y&X& X&Z&Y\cr I& I&I&I& X&Y&Z&X&Z&Y&
X&Y&Z& X&Y&Z\cr I& I&I&I& Y&Z&X&Y&X&Z&       Y&Z&X& Y&Z&X\cr
\hline\hline
\multicolumn{10}{c}{[10]=[[10,4,3]]}&\multicolumn{6}{c}{[[6,0,4]]}
\end{array}\\ \\
\begin{array}{c@{\hskip 3pt}c@{\hskip 3pt}c@{\hskip 3pt}c@{\hskip 3pt}c@{\hskip 3pt}
c@{\hskip 3pt}c@{\hskip 3pt}c@{\hskip 3pt}c@{\hskip 3pt}cc@{\hskip
3pt} c@{\hskip 3pt}c@{\hskip 3pt}c@{\hskip 3pt}c@{\hskip 3pt}c}
\hline\hline X& X&X&X&       X&X&X& X&X&X& X&X&X&X&X&X\cr Z& Z&Z&Z&
Z&Z&Z& Z&Z&Z& Z&Z&Z&Z&Z&Z\cr I& X&Y&Z&       Y&X&Z& Z&Y&X&
I&I&I&X&Y&Z\cr I& Y&Z&X&       Z&Y&X& X&Z&Y& I&I&I&Y&Z&X\cr I&
I&I&I&       X&Y&Z& X&Y&Z& X&Y&Z&X&Z&Y\cr I& I&I&I&       Y&Z&X&
Y&Z&X& Y&Z&X&Y&X&Z\cr \hline\hline
\multicolumn{10}{c}{[10,6]_1}&\multicolumn{6}{c}{[6,6]_1}\end{array}\\ \\
\begin{array}{c@{\hskip
3pt}c@{\hskip 3pt}c@{\hskip 3pt}c@{\hskip 3pt}c@{\hskip 3pt}
c@{\hskip 3pt}c@{\hskip 3pt}c@{\hskip 3pt}c@{\hskip 3pt}cc@{\hskip
3pt} c@{\hskip 3pt}c@{\hskip 3pt}c@{\hskip 3pt}c@{\hskip 3pt}c}
\hline \hline X&X&X&X& X&X&X& X&X&X& X&X&X& X&X&X\cr Z&Z&Z&Z& Z&Z&Z&
Z&Z&Z& Z&Z&Z& Z&Z&Z\cr I&X&Y&Z& I&I&I& X&Y&Z& Y&Z&X& Z&X&Y\cr
I&Y&Z&X& I&I&I& Y&Z&X& Z&X&Y& X&Y&Z\cr I&I&I&I& X&Y&Z& X&Y&Z& X&Y&Z&
X&Y&Z\cr I&I&I&I& Y&Z&X& Y&Z&X& Y&Z&X& Y&Z&X\cr \hline\hline
\multicolumn{10}{c}{[10,6]_2}&\multicolumn{6}{c}{[6,6]_2}
\end{array}
\end{array}
\end{equation*}
\end{table}

\begin{table}[h]\label{edd} \caption{Further constructions of 2ed-blocks.}
\begin{equation*}
\renewcommand{\arraystretch}{1.3}
\begin{array}{cc}
\begin{array}{c@{\hskip 3pt}c@{\hskip 3pt}c@{\hskip 3pt}c@{\hskip 3pt}c@{\hskip 3pt}c@{\hskip 3pt}c@{\hskip 3pt}c@{\hskip 3pt}c@{\hskip 3pt}c@{\hskip 3pt}c@{\hskip 3pt}c@{\hskip 3pt}c@{\hskip 3pt}c@{\hskip 3pt}c@{\hskip 3pt}c@{\hskip 3pt}c@{\hskip 3pt}c@{\hskip 3pt}c@{\hskip 3pt}c@{\hskip 3pt}c@{\hskip 3pt}c@{\hskip 3pt}c@{\hskip 3pt}c@{\hskip 3pt}c@{\hskip 3pt}c@{\hskip 3pt}c@{\hskip 3pt}c
c@{\hskip 3pt}c@{\hskip 3pt}c@{\hskip 3pt}c@{\hskip 3pt}c@{\hskip
3pt}c@{\hskip 3pt}ccc} \hline\hline
[5,5]_2&[5,5]_2&[5,5]_2&[3,5]_2\\
I(5)&X(5)&Y(5)&Z(3)\\
I(5)&Y(5)&Z(5)&X(3)\\
\hline\hline
\multicolumn{4}{c}{[18,7]_1}
\end{array}&
\begin{array}{c@{\hskip 3pt}c@{\hskip 3pt}c@{\hskip 3pt}c@{\hskip 3pt}c@{\hskip 3pt}c@{\hskip 3pt}c@{\hskip 3pt}c@{\hskip 3pt}c@{\hskip 3pt}c@{\hskip 3pt}c@{\hskip 3pt}c@{\hskip 3pt}c@{\hskip 3pt}c@{\hskip 3pt}c@{\hskip 3pt}c@{\hskip 3pt}c@{\hskip 3pt}c@{\hskip 3pt}c@{\hskip 3pt}c@{\hskip 3pt}c@{\hskip 3pt}c@{\hskip 3pt}c@{\hskip 3pt}c@{\hskip 3pt}c@{\hskip 3pt}c@{\hskip 3pt}c@{\hskip 3pt}c
c@{\hskip 3pt}c@{\hskip 3pt}c@{\hskip 3pt}c@{\hskip 3pt}c@{\hskip
3pt}c@{\hskip 3pt}ccc} \hline\hline
[7,5]_1&[5,5]_2&[3,5]_2&[3,5]_2\\
I(7)&X(5)&Y(3)&Z(3)\\
I(7)&Y(5)&Z(3)&X(3)\\
\hline\hline\multicolumn{4}{c}{[18,7]_2}
\end{array}\\\\
\begin{array}{c@{\hskip 3pt}c@{\hskip 3pt}c@{\hskip 3pt}c@{\hskip 3pt}c@{\hskip 3pt}c@{\hskip 3pt}c@{\hskip 3pt}c@{\hskip 3pt}c@{\hskip 3pt}c@{\hskip 3pt}c@{\hskip 3pt}c@{\hskip 3pt}c@{\hskip 3pt}c@{\hskip 3pt}c@{\hskip 3pt}c@{\hskip 3pt}c@{\hskip 3pt}c@{\hskip 3pt}c@{\hskip 3pt}c@{\hskip 3pt}c@{\hskip 3pt}c@{\hskip 3pt}c@{\hskip 3pt}c@{\hskip 3pt}c@{\hskip 3pt}c@{\hskip 3pt}c@{\hskip 3pt}c
c@{\hskip 3pt}c@{\hskip 3pt}c@{\hskip 3pt}c@{\hskip 3pt}c@{\hskip
3pt}c@{\hskip 3pt}ccc} \hline\hline
[7,5]_2&[5,5]_2&[5,5]_2&[3,5]_2\\
I(7)&X(5)&Y(5)&Z(3)\\
I(7)&Y(5)&Z(5)&X(3)\\
\hline\hline\multicolumn{4}{c}{[20,7]_2}
\end{array}&
\begin{array}{c@{\hskip 3pt}c@{\hskip 3pt}c@{\hskip 3pt}c@{\hskip 3pt}c@{\hskip 3pt}c@{\hskip 3pt}c@{\hskip 3pt}c@{\hskip 3pt}c@{\hskip 3pt}c@{\hskip 3pt}c@{\hskip 3pt}c@{\hskip 3pt}c@{\hskip 3pt}c@{\hskip 3pt}c@{\hskip 3pt}c@{\hskip 3pt}c@{\hskip 3pt}c@{\hskip 3pt}c@{\hskip 3pt}c@{\hskip 3pt}c@{\hskip 3pt}c@{\hskip 3pt}c@{\hskip 3pt}c@{\hskip 3pt}c@{\hskip 3pt}c@{\hskip 3pt}c@{\hskip 3pt}c
c@{\hskip 3pt}c@{\hskip 3pt}c@{\hskip 3pt}c@{\hskip 3pt}c@{\hskip
3pt}c@{\hskip 3pt}ccc} \hline\hline
[7,5]_2&[7,5]_2&[7,5]_2&[5,5]_2\\
I(7)&X(7)&Y(7)&Z(5)\\
I(7)&Y(7)&Z(7)&X(5)\\
\hline\hline\multicolumn{4}{c}{[26,7]_2}
\end{array}
\end{array}
\end{equation*}
\end{table}

\section{Discussions}

We have described a general construction of all the optimal stabilizer
codes of distance $3$ for lengths $n> 38$ by pasting known codes and
a special construction of the optimal pure stabilizer codes of length
$5\le n\le 38$ case by case by employing a generalization of the
stabilizer pasting to noncommuting sets of stabilizers, i.e.,
2ed-blocks pasting. For three families of lengths we have worked out
analytically the linear programming bound, which is strictly
stronger than the quantum Hamming bound and ensures the optimality
of our codes for these lengths. for all lengths except $n=6$ there are pure optimal codes.

Apparently the construction given by Theorem 2 is not unique.
Firstly there are different constructions for the optimal code
$[2^m]$ \cite{cal2}. Secondly there are other constructions such as
\begin{eqnarray}\label{caa}
[8\cdot(2^{2m-1}-\alpha_{1})]\rhd [8\cdot (2^{2m-3}-\alpha_{2})]
\rhd \ldots \mbox{\hskip 20pt   }\cr \ldots\rhd [8\cdot (2^3-\alpha_{m-1})]\rhd [17-\beta]
\end{eqnarray}
or
\begin{eqnarray}
\label{cba}
[8\cdot(2^{2m}-\alpha_{1})]\rhd [8\cdot (2^{2m-2}-\alpha_{2})]
\rhd \ldots\mbox{\hskip 22pt   }\cr\ldots \rhd [8\cdot (2^4-\alpha_{m-1})]\rhd [38-\beta]
\end{eqnarray}
where $\alpha_i+3\le 2^{2(m-i+1)-1}$ or $2^{2(m-i+1)}$ respectively
and $\alpha=\sum_{i=1}^{m-1}\alpha_i.$ For different choices of
$\{\alpha_i\}$ the resulting codes may be inequivalent. This raises the problem
of the classification of the optimal codes. Finally
our approach should turn out to be useful to investigate nonbinary codes (see Ref.\cite{BE}) as well.

\begin{remark}
At time of finishing the first version of this paper the optimal codes of lengths
$n=36,37,38,81$, which have been constructed in Ref.\cite{jb}
have been missing in Grassl's code table.
\end{remark}

\end{document}